\documentclass[conference]{IEEEtran}

\usepackage{tikz}
\usetikzlibrary{arrows}
\usetikzlibrary{calc}
\usetikzlibrary{snakes}
\usepackage{cite}
\usepackage{graphicx}
\usepackage[cmex10]{amsmath}
\usepackage{amsfonts}
\usepackage{amsthm}
\usepackage{amssymb}
\usepackage{algorithmic}
\usepackage{array}
\usepackage{mdwmath}
\usepackage{mdwtab}
\usepackage{eqparbox}
\usepackage{fixltx2e}
\usepackage{url}
\usepackage{comment}
\IEEEoverridecommandlockouts

\newtheorem{definition}{Definition}
\newtheorem{theorem}{Theorem}
\newtheorem{lemma}{Lemma}

\newcommand{\remove}[1]{}

\begin{document}

\title{How to Securely Compute the Modulo-Two Sum of Binary Sources}

\author{\IEEEauthorblockN{Deepesh Data}
\IEEEauthorblockA{
TIFR, Mumbai\\
deepeshd@tifr.res.in}
\and
\IEEEauthorblockN{Bikash K. Dey}
\IEEEauthorblockA{IIT Bombay, Mumbai\\
bikash@ee.iitb.ac.in}
\and
\IEEEauthorblockN{Manoj Mishra}
\IEEEauthorblockA{IIT Bombay, Mumbai\\
mmishra@ee.iitb.ac.in}
\and
\IEEEauthorblockN{Vinod M. Prabhakaran}
\IEEEauthorblockA{TIFR, Mumbai\\
vinodmp@tifr.res.in}\thanks{Authors are listed in the alphabetical order.}}

\maketitle

\begin{abstract}

In secure multiparty computation, mutually distrusting users in a network want
to collaborate to compute functions of data which is distributed among the
users. The users should not learn any additional information about the data of
others than what they may infer from their own data and the functions they are
computing. Previous works have mostly considered the worst case
context (i.e., without assuming any distribution for the data);~Lee and Abbe
(2014) is a notable exception. Here, we study the average case (i.e., we work
with a distribution on the data) where correctness and privacy is only desired
asymptotically. 

For concreteness and simplicity, we consider a secure version of the function
computation problem of K\"orner and Marton (1979) where two users observe a
doubly symmetric binary source with parameter $p$ and the third user wants to
compute the XOR. We show that the amount of communication and randomness
resources required depends on the level of correctness desired. When zero-error
and perfect privacy are required, the results of Data et al. (2014) show that it
can be achieved if and only if a total rate of 1 bit is communicated between every
pair of users and private randomness at the rate of 1 is used up.  In contrast,
we show here that, if we only want the probability of error to vanish
asymptotically in blocklength, it can be achieved by a lower rate (binary
entropy of $p$) for all the links and for private randomness; this also
guarantees perfect privacy. We also show that no smaller rates are possible
even if privacy is only required asymptotically.

\end{abstract}

\IEEEpeerreviewmaketitle

\section{Introduction}\label{sec:intro}

In secure multiparty computation (MPC), mutually distrusting users in a network
want to collaborate to compute functions of data which is distributed among the
users.  The users should not learn any additional information about the data of
others than what they may infer from their own data and the functions they are
computing.  Various applications such as online auctions, electronic voting,
and privacy preserving data mining motivate the study of MPC~\cite[Chapter
1]{CramerDN}.

In a seminal result, Ben-Or, Goldwasser, and Wigderson~\cite{BGW88} (also see
Chaum, Cr\'epeau, and Damg{\aa}rd~\cite{CCD88}), established that information
theoretically secure computation of any function is feasible by $n$ users who
are connected pairwise by private noiseless communication links and who have
access to private randomness, even if any set of strictly less than $n/2$
users collude. The $n/2$ threshold is for the honest-but-curious setting where
the users do not deviate from the protocol during its execution, but a subset
of users may collude at the end of the protocol to try to infer information
about data of other users that they cannot infer from their own data and
outputs of the function they computed.  In making this inference, they may make
use of their own data, their private randomness, and all the messages they sent
and received during the execution of the protocol. The threshold is $n/3$ for
the malicious case where the colluding users may also deviate from the protocol
during its execution. It is also known that these thresholds are tight in the
sense that there exist functions which cannot be securely computed when the
number of colluders exceed these thresholds\footnote{In cases where the number
of colluders exceed these thresholds, additional noisy resources (e.g.,
distributed sources or noisy channels) can be exploited to perform secure
computation~\cite{CrepeauK88}. In this paper, our focus is on the case where
such additional resources are unavailable.}.
 
The amount of communication and randomness required to securely compute in
the model of~\cite{BGW88,CCD88} is an important open problem. Several works
have addressed this to a limited extent, for the most part, in the worst
case context (i.e., without assuming any distribution for the data) and for
zero-error computation with perfect
privacy~\cite{Kushilevitz89,FranklinY92,ChorK93,FeigeKN94,KushilevitzM97,BlundoSPV99,GalR05,DataPP14}.
{Most directly relevant to this paper is~\cite{DataPP14} where generic
information theoretic lower bounds were obtained for zero-error computation
in a three-user model with perfect privacy against individual users.} 

In this paper, in contrast to the above works, we take a distributed source
coding approach to this problem. Specifically, we will assume a probability
distribution for the data (discrete memoryless distributed source), and seek
the average-case performance under asymptotically vanishing error and vanishing
privacy leakage. We would like to point out that~\cite{LeeAbbe14} already
considered a similar setting, but for a much weaker notion of security than
what we consider below. For concreteness and simplicity, we focus on the famous
example of K\"orner and Marton~\cite{KornerM79}. Consider
Figure~\ref{fig:xor-setup}.  Alice (user 1) and Bob (user 2) observe data $X^n$
and $Y^n$ which are $n$-length bit strings drawn i.i.d. according to the
distribution $p_{XY}(x,y)= \frac{p}{2}1_{x\neq y} + \frac{1-p}{2}1_{x=y}$,
where $0\leq p \leq 1/2$. This is sometimes referred to as the doubly symmetric
binary source (DSBS) with parameter $p$. Charlie (user 3) wants to compute the
function $Z^n=X^n\oplus Y^n$, the binary sum (XOR) of the corresponding
elements of the data vectors. Note that $Z^n\sim$~i.i.d.  Bernoulli($p$).
K\"orner and Marton gave a function computation scheme which requires a rate of
$R=H_2(p)$ each from Alice to Charlie and Bob to Charlie such that Charlie
recovers $Z^n$ with vanishing error (as $n\rightarrow \infty$), where $H_2$ is
the binary entropy function.  The scheme involved Alice and Bob sending
syndromes of their observations computed for the same capacity approaching
linear code for binary symmetric channel with crossover probability $p$
(BSC($p$)). Charlie computes the binary sum of the syndromes to obtain the
syndrome of $Z^n$ from which $Z^n$ can be recovered with high probability.

We will additionally require the privacy conditions that Alice and Bob must not
learn more information about each other's data than what they can already infer
from their own data, and that Charlie should not learn more information about
Alice and Bob's data than what he can infer from the binary sum $Z^n$ he wants
to compute. Users only have access to private randomness and pairwise noiseless
bidirectional communication links which they may use over multiple rounds.  The
users are assumed to be honest-but-curious.  By~\cite{BGW88}, it is known that
any function of the data at Alice and Bob can be computed at Charlie while
guaranteeing these privacy requirements.  We are interested in characterizing
the rates of communication (expected number of bits exchanged over each link
per source symbol) and the rate of private randomness used. Our main result is
a characterization of these rates for the case where we only require that
Charlie reconstruct $Z^n$ with asymptotically vanishing probability of error
(as $n\rightarrow \infty$) and when the privacy conditions hold in the sense of
asymptotically vanishing information leakage (stated formally in
Section~\ref{sec:problem}).

One of the examples in~\cite{DataPP14} gives the answers for the zero-error
and perfect privacy case. It is easy to see that a simple protocol achieves
a rate of one bit per source symbol over each of the links and a rate of
one bit of private randomness\footnote{For example, Charlie sends to Alice
an $n$-length vector $K^n$ of i.i.d. uniformly distributed bits from his
private randomness; Alice sends $K^n\oplus X^n$ to Bob; Bob in turn sends
$(K^n\oplus X^n)\oplus Y^n$ to Charlie from which he can recover $X^n\oplus
Y^n$. Perfect privacy is easy to verify.}. \cite{DataPP14} shows that there
is no zero-error, perfectly private protocol which can do with less.  In
fact, none of these rates can be lowered even at the expense of higher
rates for the others. For completeness, a short proof of this is presented
in the appendix.

If the zero-error requirement is relaxed to vanishing error, the coding scheme
of K\"orner and Marton suggests the following secure computation scheme which
only requires rates of $H_2(p)$. Recall that K\"orner and Marton's function
computation scheme requires a rate of $R=H_2(p)$ from each of Alice and Bob to
Charlie. For secure computation, Alice sends to Bob an $nR$-length vector
$K^{nR}$ of i.i.d. uniformly distributed bits drawn from her private
randomness. Both Alice and Bob send their respective syndromes (of length $nR$)
XOR-ed with $K^{nR}$ to Charlie. Charlie adds these to recover the syndrome of
$Z^n$ as before. It is easy to see that this scheme, in fact, guarantees
perfect privacy. We show that this scheme is optimal in the sense that none of
the rates can be reduced even at the expense of higher rates for the others and
even if only asymptotically vanishing information leakage is desired. We prove
this converse result for a fairly general class of interactive protocols.

Related works include works on function computation without the privacy
requirement~\cite{OrlitskyR01,MaI11,MaI12}. As already pointed out above,
another related work is~\cite{LeeAbbe14}. It studies the randomness required
for secure sum computation under two different settings: (i) in the zero-error,
perfect privacy, worst-case setting, and (ii) average case, asymptotically
correct setting under a much weaker notion of privacy that users are unable to
asymptotically correctly guess the entire data of another user, but when no
private randomness is available to the users.

\section{Problem Definition and Statement of Results}
\label{sec:problem}
\usetikzlibrary{decorations.markings}
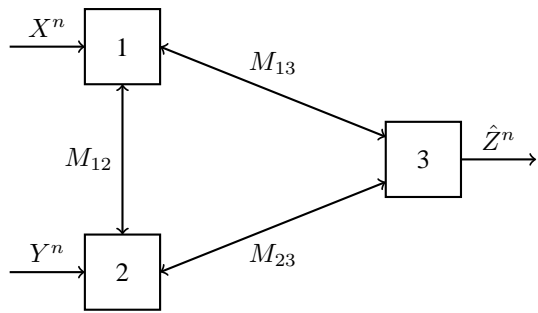
\begin{figure}[h]
\setlength{\unitlength}{1cm}
\centering
\begin{tikzpicture}[scale=1, thick]
\draw (1,0) rectangle (2,1);
\draw (1,3) rectangle (2,4);
\draw (5,1.5) rectangle (6,2.5);

\draw [->] (0,0.5) -- (1,0.5);
\draw [->] (0,3.5) -- (1,3.5);
\draw [->] (6,2) -- (7,2);
\draw [<->] (1.5,1) -- (1.5,3);
\draw [<->] (2,0.5) -- (5,1.7);
\draw [<->] (2,3.5) -- (5,2.3);

\node at (1.5,0.5) {2};
\node at (1.5,3.5) {1};
\node at (5.5,2) {3};

\node [above] at (0.5,0.5) {$Y^n$};
\node [above] at (0.5,3.5) {$X^n$};
\node [above] at (6.5,2) {$\hat{Z}^n$};

\node [left] at (1.5,2) {$M_{12}$};
\node [above] at (3.5,3) {$M_{13}$};
\node [below] at (3.5,1) {$M_{23}$};

\end{tikzpicture}
\caption{Setup for computing XOR securely}
\label{fig:xor-setup}
\end{figure}

In the setup of Figure~\ref{fig:xor-setup}, Alice (user 1) and Bob (user 2)
have blocks of data/input bits $X^n$ and $Y^n$ respectively, where ($X^n,Y^n$)
are drawn i.i.d. from a Doubly Symmetric Binary Source (DSBS)-$p$ distribution
$p_{XY}(x,y)$ such that $X$ and $Y$ are both $\text{Bernoulli}(1/2)$, and
$Pr(X\neq Y) = p$.  Charlie (user 3) wishes to compute an estimate $\hat{Z}^n$
of the bit-wise XOR of $X^n$ and $Y^n$.  That is, $\hat{Z}^n$ is an estimate of
$Z^n = X^n \oplus Y^n$. Notice that $X^n,Z^n$ are independent and so are
$Y^n,Z^n$. Each pair of users is connected by a binary, error-free,
bidirectional link private from the other user.  At the beginning of the
protocol, all users are allowed to generate private random variables, i.e.,
they may generate random variables which are independent of each other and the
data. We are interested in reliably and \emph{securely computing} the XOR,
where any single user does not learn anything about the other users'
inputs/output (if any) at the end of the protocol than what its own
input/output (if any) reveals about them. We formalise this in Definition
\ref{defn:ach_rate_triple}. 
We assume that the users are honest-but-curious, i.e., they follow the protocol
honestly but are interested in obtaining additional illegitimate information
about the inputs/output of other users from all the messages exchanged.

To accomplish the above task, users need to communicate. Communication proceeds
over multiple rounds. In each round~$t$, every user sends a (potentially empty)
message in the form of a variable length, binary string to every other user.
Let $M_{\overrightarrow{ij},t}$ denote the message from user~$i$ to user~$j$,
sent in round~$t$. $M_{\overrightarrow{ij},t}$ may depend only on user~$i$'s input (if any), private
randomness, and all the messages it has seen so far. We require that
$M_{\overrightarrow{ij},t}$ belong to a variable length prefix-free code
${\mathcal C}_{\overrightarrow{ij},t}$ which itself could be random (determined
by the inputs and the private randomnesses).  However, at the beginning of
round-$t$, both users $i$ and $j$ must each deterministically know ${\mathcal
C}_{\overrightarrow{ij},t}$ 
%from their data, private randomness, and messages they have received so far. 
from the messages they have exchanged with each other over the $ij$-link
between them in the previous $(t-1)$ rounds. The total number of rounds is also
allowed to be random, but from the above description, it is clear that each
user will come to know when the exchanges involving it have finished.  We
insist that the protocol terminates in finite number of rounds with probability
1. On termination, Charlie outputs $\hat{Z}^n$ as a function of his private
randomness and all the messages he received.

\begin{definition}
\label{defn:protocol}
In a \emph{protocol} $\Pi_{n}$, where $n$ is the input block length, users
exchange messages with each other over several rounds as described above at
the end of which Charlie produces an output $\hat{Z}^n$.
\end{definition}

We use the following notation throughout this paper.
The \emph{transcript} on $ij$-link at time $t$ is 
\begin{equation*}
M_{ij,t} := (M_{\overrightarrow{ij},t}, M_{\overrightarrow{ji},t}).
\end{equation*}
We also define $M_{ij}^t:= (M_{ij,\tau})_{\tau=1}^t$, and 
$M_{ij}:=M_{ij}^\infty$ denotes the final transcript on the $ij$-link. Finally,
$L_{\overrightarrow{ij},t}$ is the length, in bits, of the message
$M_{\overrightarrow{ij},t}$. Clearly, $L_{\overrightarrow{ij},t}$ is a random
variable and $L_{\overrightarrow{ij},t} \in \{0, 1, 2, \ldots\}$. Similarly 
the length random variables $L_{ij,t}$, $L_{ij}^t$ and $L_{ij}$ are defined 
as the lengths of $M_{ij,t}$, $M_{ij}^t$ and $M_{ij}$ respectively.

\begin{definition}
The \emph{rate} of $\Pi_{n}$ is defined by the quadruple ($r_{13,n},r_{23,n},r_{12,n}, \rho_n$) where:
\begin{align*}
 r_{13,n} & := \frac{1}{n} \mathbb{E}[L_{13}]  \\
 r_{23,n} & := \frac{1}{n} \mathbb{E}[L_{23}]  \\
 r_{12,n} & := \frac{1}{n} \mathbb{E}[L_{12}]  \\
 \rho_n & := \frac{1}{n}H(M_{13},M_{23},M_{12}|X^n,Y^n)
\end{align*}
\end{definition}

We note that once the protocol ends at some finite time, all the subsequent messages are of zero-length.

\begin{definition}
\label{defn:ach_rate_triple}
A rate quadruple ($R_{13},R_{23},R_{12}, \rho$) is achievable in the setup
of Figure~\ref{fig:xor-setup} if there exists a sequence of protocols
$(\Pi_{n})_{n \in \mathbb{N}}$, with rates
$r_{ij, n} \leq R_{ij}$ for $i,j=1,2,3, i\neq j$, and $\rho_n \leq \rho$,
such that
\begin{align}
 P(\hat{Z}^n \neq Z^n) & \longrightarrow 0,
 \label{eqn:ach_rate_1} \\
 I(M_{13},M_{12} ; Y^n | X^n) & \longrightarrow 0,
 \label{eqn:ach_rate_2} \\
 I(M_{23},M_{12} ; X^n | Y^n) & \longrightarrow 0,
 \label{eqn:ach_rate_3}\\
 I(M_{13}, M_{23} ; X^n,Y^n | Z^n) & \longrightarrow 0.
 \label{eqn:ach_rate_4}
\end{align}
\end{definition}
Notice that we do {\em not} need to explicitly include the private random
variables in the privacy conditions since conditioned on the messages and
input (if any) at a user, its private random variable is independent of the
the other input(s). \eqref{eqn:ach_rate_4} is a privacy promise to
Alice and Bob that Charlie learns only asymptotically vanishing amount of
information about their data in addition to $Z^n$ which he is allowed to
compute. Similar interpretations hold for the other two privacy conditions. 

\begin{definition}
The \emph{rate region} $\mathcal{R}$ for the setup in 
Figure~\ref{fig:xor-setup} is defined as the closure of the set of 
all achievable rate quadruples.
\end{definition}

Our main result is a characterization of the rate region $\mathcal{R}$.
\begin{theorem}\label{theorem:main_thm}
\begin{align*}
\mathcal{R} = \{(R_{13},R_{23},R_{12}) : & \min(R_{13},R_{23}, R_{12},
\rho) \geq H(Z) \}.
\end{align*}
\end{theorem}
{\em Remark 1:} The achievability is in fact proved for the perfect privacy case where the
privacy conditions \eqref{eqn:ach_rate_2}-\eqref{eqn:ach_rate_4} hold with
equality. And, our converse is proved for the weak privacy setting where
\eqref{eqn:ach_rate_2}-\eqref{eqn:ach_rate_4} are replaced by
\eqref{eq:weakach_rate_2}-\eqref{eq:weakach_rate_4} (see
Section~\ref{sec:converse}), i.e., only the rates of information leaked need
to vanish asymptotically.\\
%\begin{align}
%\frac{1}{n} I(M_{13},M_{12} ; Y^n | X^n)  &\longrightarrow 0, \label{eq:weakach_rate_2} \\
%\frac{1}{n} I(M_{23},M_{12} ; X^n | Y^n)  &\longrightarrow 0, \label{eq:weakach_rate_3}\\
%\frac{1}{n} I(M_{13}, M_{23} ; X^n,Y^n | Z^n) &\longrightarrow 0. \label{eq:weakach_rate_4}
%\end{align}
{\em Remark 2:} We note that if Charlie is required to compute $Z^n$ with
zero error and perfect privacy (i.e., when \eqref{eqn:ach_rate_1}-\eqref{eqn:ach_rate_4} hold with
equality), then on all three links we need $n$ bits to be exchanged and
$n$ bits of private randomness is needed~\cite{DataPP14}. This result is
discussed in the Appendix. 

\section{Proof of Achievability}
\label{sec:ach-protocol}

Our achievability scheme directly builds on K\"orner and Marton's scheme for
modulo-two sum of doubly symmetric binary sources~\cite{KornerM79}. Since
$(X,Y)$ is a DSBS-$p$, their XOR $Z=X\oplus Y$ is Bernoulli($p$).  It is
well-known that linear codes achieve the capacity of the binary symmetric
channel. i.e., for fixed $\epsilon>0$, $R=H(p)+\epsilon$ and for each block
length $n$, there is a linear coding matrix $\Lambda_n$ of size $(nR)\times
n$ and a decoder $\mathcal{D}_n$ such that $P(\mathcal{D}_n(\Lambda_n
Z^n)\neq Z^n)\to 0$ as $n\to \infty$.
%For a fixed $\epsilon >0$, $R=H(p)+\epsilon$, a large block length $n$, and
%$m:=nR$, let $\Lambda_{n}$ be a $m \times n$ random matrix with i.i.d.
%$Bern(1/2)$ entries. It is well known that, if $Z^n$ is encoded as $\Lambda_n
%Z^n$, then a typicality decoder $\mathcal{D}_n$ recovers $Z^n$ from $\Lambda_n
%Z^n$ with probability of error converging to zero as $n\rightarrow \infty$. 
In K\"orner and Marton's scheme, Alice sends $(\Lambda_n X^n)$ and Bob sends
$(\Lambda_n Y^n)$ to Charlie, who XORs the received vectors component-wise to
get $(\Lambda_n Z^n)$. Using the decoder $\mathcal{D}_n$, Charlie recovers
$Z^n$ with vanishing probability of error.
  
In our scheme, Alice first generates $m := nR$ private random Bernoulli($1/2$) bits
$K^{m}$ and sends it to Bob.  She also sends $A=K^m \oplus (\Lambda_n
X^n)$ to Charlie. Bob sends $B=K^m \oplus (\Lambda_n Y^n)$ to Charlie.
Charlie XORs the two binary vectors he received component-wise to get
$(\Lambda_n Z^n)$ and proceeds to decode as before.  This scheme has the
rate-tuple $(R,R,R,R)$ with $R=H(p)+\epsilon$.  Since $\epsilon$ can be
chosen to be arbitrarily small, it is sufficient to consider this class of
protocols for the achievability of Theorem~\ref{theorem:main_thm}.

It is straightforward to show that our protocol is perfectly private, i.e.,
(\ref{eqn:ach_rate_2}), (\ref{eqn:ach_rate_3}), and (\ref{eqn:ach_rate_4})
hold with equality. For (\ref{eqn:ach_rate_2}),
\begin{align*}
I(A,K^m ; Y^n | X^n) 
    & =  I(K^m ; Y^n | X^n)  + I(A ; Y^n | X^n,K^m)  =  0,
\end{align*}
since $K^m$ is independent of ($X^n,Y^n$), and $A$ is a function of ($X^n, K^m$).
Similarly, (\ref{eqn:ach_rate_3}) holds with equality.
%\begin{align*}
%I(B, K^m ; X^n | Y^n) 
%    & =  I(K^m ; X^n | Y^n)  + I(B ; X^n | Y^n, K^m) \\
%    & =  0.
%\end{align*}
Finally, for (\ref{eqn:ach_rate_4}),
\begin{align*}
I(&A, B; X^n,Y^n | Z^n) \\
& = I(A, B; X^n | Z^n) \\
& = I(A, B,Z^n; X^n ) - \underbrace{I(Z^n;X^n)}_{=\ 0} \\
& = I(K^m \oplus (\Lambda_nX^n), K^m \oplus (\Lambda_nY^n),Z^n ; X^n ) \\
& = I(K^m \oplus (\Lambda_nX^n), Z^n ; X^n) \\
& = 0,
\end{align*}
since $(K^m \oplus \Lambda_nX^n, Z^n)$ is independent of $X^n$. The
penultimate step follows from the fact that $K^m \oplus (\Lambda_nY^n) = (K^m
\oplus (\Lambda_nX^n)) \oplus (\Lambda_nZ^n)$.

\section{Proof of Converse}
\label{sec:converse}
Let ($R_{13},R_{23},R_{12},\rho$) be an achievable rate quadruple. Then, by
Definition~\ref{defn:ach_rate_triple}, there exists a sequence of protocols
$(\Pi_{n})_{n \in \mathbb{N}}$ with the corresponding rates
$r_{ij,n} \leq R_{ij}$, $i,j = 1,2,3, i\neq j$, $\rho_n \leq \rho$,
satisfying  \eqref{eqn:ach_rate_1} and the
weak privacy conditions
\begin{align}
\epsilon_1 & := \frac{1}{n}\, I(M_{13},M_{12} ; Y^n | X^n)  \longrightarrow 0, \label{eq:weakach_rate_2} \\
\epsilon_2 & := \frac{1}{n}\, I(M_{23},M_{12} ; X^n | Y^n)  \longrightarrow 0, \label{eq:weakach_rate_3}\\
\epsilon_3 & := \frac{1}{n}\, I(M_{13}, M_{23} ; X^n,Y^n | Z^n)
\longrightarrow 0, \label{eq:weakach_rate_4}\\
\intertext{as $n\to \infty$. By Fano's inequality, \eqref{eqn:ach_rate_1}
implies, as $n\to \infty$,}
\epsilon_4 & := \frac{1}{n} H(Z^n|\hat{Z}^n) \longrightarrow 0.
\end{align}
For the lower bound on $R_{12}$, we proceed as follows.
\begin{align}
&\mathbb{E}[L_{12}]\\
&= \mathbb{E}\left[\sum_{t=1}^\infty L_{\overrightarrow{12},t} +
L_{\overrightarrow{21},t}\right] \nonumber \\
&= \sum_{t=1}^\infty \mathbb{E}\left[L_{\overrightarrow{12},t}\right] +
\mathbb{E}\left[L_{\overrightarrow{21},t}\right] \nonumber \\
&\geq \sum_{t=1}^\infty H(M_{\overrightarrow{12},t}|{\mathcal
C}_{\overrightarrow{12},t}) + H(M_{\overrightarrow{21},t}|{\mathcal
C}_{\overrightarrow{21},t}) \label{eq:entropy_lb} \\
%&\geq \sum_{t=1}^\infty H(M_{\overrightarrow{12},t}|M_{12}^{t-1},M_{13}) + H(M_{\overrightarrow{21},t}|M_{12}^{t-1},M_{13}) \label{eq:M12_first}\\
&\geq  \sum_{t=1}^\infty H(M_{\overrightarrow{12},t}|M_{12}^{t-1}) + H(M_{\overrightarrow{21},t}|M_{12}^{t-1}) \label{eq:M12_first}\\
%&\geq \sum_{t=1}^{\infty} H(M_{\overrightarrow{12},t},M_{\overrightarrow{21},t} | M_{12}^{t-1},M_{13}) \nonumber \\
&\geq \sum_{t=1}^{\infty} H(M_{\overrightarrow{12},t},M_{\overrightarrow{21},t} | M_{12}^{t-1}) \nonumber \\ &= H(M_{12})\\
%&= H(M_{12}|M_{13}) \nonumber \\
&\geq H(M_{12}|M_{13}) \nonumber \\
&\geq I(X^n;M_{12} | M_{13}) \nonumber \\
&= I(X^n;M_{12},M_{13}) - I(X^n;M_{13}) \nonumber \\
&\geq I(X^n;M_{13}|M_{12}) - n\epsilon_3 \label{eq:M12_XM13}\\
&= H(X^n|M_{12}) - H(X^n|M_{12},M_{13}) - n\epsilon_3  \nonumber \\
&= H(X^n|M_{12}) - I(X^n;Y^n|M_{12},M_{13}) \nonumber \\
&\hspace{3.1cm} - H(X^n|Y^n,M_{12},M_{13}) - n\epsilon_3  \nonumber \\
&\geq H(X^n|M_{12}) - I(X^n;Y^n|M_{12},M_{13}) - n\epsilon_4 - n\epsilon_3  \label{eq:M12_second} \\
&= H(X^n) - I(X^n;M_{12}) - I(X^n;Y^n|M_{12},M_{13}) \nonumber \\
& \hspace{5.8cm} - n\epsilon_4 - n\epsilon_3 \nonumber \\
&= \underbrace{H(X^n|Y^n)}_{=\ nH(Z)} + I(X^n;Y^n) - I(X^n;M_{12}) \nonumber \\
&\hspace{3.1cm} - \underbrace{I(X^n;Y^n|M_{12},M_{13})}_{\leq I(X^n,M_{13};Y^n|M_{12})} - n\epsilon_4 - n\epsilon_3 \nonumber \\
&\geq nH(Z) + I(X^n;Y^n) - I(X^n;M_{12}) - I(X^n;Y^n|M_{12}) \nonumber \\
&\hspace{3.1cm} - \underbrace{I(M_{13};Y^n|X^n,M_{12})}_{\leq\ n\epsilon_1, \text{ by }\eqref{eq:weakach_rate_2}} - n\epsilon_4 - n\epsilon_3 \nonumber \\
&= nH(Z) + I(X^n;Y^n) - I(X^n;Y^n,M_{12}) \nonumber \\
&\hspace{5.6cm} - n\epsilon_1 - n\epsilon_4 - n\epsilon_3 \nonumber \\
&= nH(Z) - \underbrace{I(X^n;M_{12}|Y^n)}_{\leq\ n\epsilon_2, \text{ by }\eqref{eq:weakach_rate_3}} - n\epsilon_1 - n\epsilon_4 - n\epsilon_3 \nonumber \\
&= nH(Z) - n\epsilon_2 - n\epsilon_1 - n\epsilon_4 - n\epsilon_3. \nonumber
\end{align}
Here, in \eqref{eq:entropy_lb}, ${\mathcal C}_{\overrightarrow{12},t}$ and ${\mathcal C}_{\overrightarrow{21},t}$ denote the prefix-free codes
that are used in sending the messages $M_{\vec{12},t}$ and $M_{\vec{21},t}$, respectively.
These codes depend on the particular instance of the protocol, and are known
to Alice and Bob based on all the messages ($M_{12}^{t-1}$) communicated between 
them till time $t-1$.
\eqref{eq:entropy_lb} follows from the fact that expected length $L$ of any
prefix-free binary code for a random variable $U$ is lower bounded by
$H(U)$~\cite[Theorem~5.3.1]{CoverT}.
%\eqref{eq:M12_first} holds because at time $t$, the prefix-free code used by any two parties (say 1 and 2) is determined by $(M_{12}^{t-1},M_{13}^{t-1})$ as well as from $(M_{12}^{t-1},M_{23}^{t-1})$. 
\eqref{eq:M12_first} holds because at time $t$, the prefix-free codes used by any two users (say 1 and 2) are determined by $M_{12}^{t-1}$. 
\eqref{eq:M12_XM13} follows because, since $X^n$ and $Z^n$ are independent,
$I(X^n;M_{13}) \leq I(X^n;M_{13},M_{23},Z^n) = I(X^n;M_{13},M_{23}|Z^n)\leq n\epsilon_3$.
\eqref{eq:M12_second} follows from $H(X^n|Y^n,M_{12},M_{13}) \leq
n\epsilon_4$ which can be seen as follows: From the cut separating Alice
from Bob and Charlie, it follows that, conditioned on $(M_{12},M_{13},Y^n)$, Charlie's output $\hat{Z}^n$ is independent of $X^n$, which implies the Markov chain $\hat{Z}^n-(M_{12},M_{13},Y^n)-X^n$. Therefore, $H(X^n|Y^n,M_{12},M_{13}) = H(X^n|Y^n,M_{12},M_{13},\hat{Z}^n)$. 
Since $Z=X\oplus Y$, we have $H(X^n|Y^n,M_{12},M_{13},\hat{Z}^n) =
H(Z^n|Y^n,M_{12},M_{13},\hat{Z}^n)\leq H(Z^n|\hat{Z}^n)=n\epsilon_4$.

Now, since $\epsilon_1 + \epsilon_2 + \epsilon_3 + \epsilon_4 \to 0$ as
$n\to \infty$, and $r_{12,n} = \frac{1}{n} \mathbb{E}[L_{12}] \leq R_{12}$, 
we have, $$R_{12} \geq H(Z).$$

The lower bound on $\mathbb{E}[L_{13}]$ and $\mathbb{E}[L_{23}]$ can be
proved along the same lines as for $\mathbb{E}[L_{12}]$.  For
$\mathbb{E}[L_{13}]$, we use the prefix free codes ${\mathcal C}_{\overrightarrow{13},t}$ and ${\mathcal C}_{\overrightarrow{31},t}$ for $M_{\vec{13},t}$ and $M_{\vec{31},t}$ at time $t$ 
%{\color{red} which can be determined from $(M_{13}^{t-1},M_{12})$}.
{\color{black} which can be determined from $(M_{13}^{t-1})$}.
Once we get to the point
%$\mathbb{E}[L_{13}] \geq H(M_{13}|M_{12})$, we apply $H(M_{13}|M_{12}) \geq
$\mathbb{E}[L_{13}] \geq H(M_{13}) \geq H(M_{13}|M_{12})$, we apply $H(M_{13}|M_{12}) \geq
I(X^n;M_{13}|M_{12})$, and from this point onwards, proceed exactly as from
\eqref{eq:M12_XM13}.
Since $\mathbb{E}[L_{13}]$ and $\mathbb{E}[L_{23}]$ are symmetric, appropriate modifications will prove the same result for $\mathbb{E}[L_{23}]$.
Thus, we have
\begin{align*}
&R_{13} \geq H(Z)\\
\text{and } &R_{23} \geq H(Z).
\end{align*}

\noindent{\em Remark:} K\"orner and Marton~\cite{KornerM79} proved a lower
bound of $H(Z)$ on $R_{13}$ and $R_{23}$ assuming non-interactive
communication between Alice/Bob and Charlie, that is, Alice and Bob both
send one message to Charlie and based on these two messages Charlie
produces the output. However, this does not directly apply here since now
there is a link between Alice and Bob, and in addition we allow interactive
communication and private randomness. Note that our bound depends on both
the privacy and correctness conditions since, in the absence of the privacy
conditions, Alice need not communicate directly with Charlie, for
instance.

For the randomness rate $\rho$, we proceed as follows:
\begin{align}
n\rho_n &\geq H(M_{12}|X^n,Y^n)\nonumber\\
&= H(M_{12}|X^n) - \underbrace{I(M_{12};Y^n|X^n)}_{\leq\ n\epsilon_1, \text{ by } \eqref{eq:weakach_rate_2}} \nonumber \\
&\geq I(M_{12};M_{13}|X^n) - n\epsilon_1 \nonumber \\
&= I(M_{12},X^n;M_{13}) - I(X^n;M_{13}) - n\epsilon_1 \nonumber \\
&\geq I(X^n;M_{13}|M_{12}) - n\epsilon_3 - n\epsilon_1 \label{eq:rand_first} \\
&\geq nH(Z) -n\epsilon_2 - n\epsilon_1 - n\epsilon_4 - n\epsilon_3 - n\epsilon_1, \label{eq:rand_second}
\end{align}
where \eqref{eq:rand_first} follows for the same reason as \eqref{eq:M12_XM13}. To bound $I(X^n;M_{13}|M_{12})$ in \eqref{eq:rand_first}, we proceed similarly as done from \eqref{eq:M12_XM13}.

Since $2\epsilon_1 + \epsilon_2 + \epsilon_3 + \epsilon_4 \to 0$ as 
$n\to \infty$, and $\rho_n \leq \rho$, we have
\begin{align*}
& \rho \geq H(Z).
\end{align*}
This completes the proof of the converse.

\noindent{\em Remark:} Our converse allows for a very
general class of protocols. We not only consider protocols with
fixed-length messages, but those with variable length messages as well. We
imposed a technical condition that the (potentially random) prefix-free
code used for any message transmission on a link be fully determined by
previous messages exchanged over the same link. Strictly
speaking, this is not necessary. It will suffice for the two communicating
users to both agree on the same code (with probability 1), but in this they
may rely on their data (if any), private randomness, and messages from
the third user as well. We believe that the same result holds even for this
slightly more general setting. The proof here can be readily extended to
derive the same lower bounds in this more general case for all but
$R_{12}$.
%{\color{red} We note that the proof presented above readily extends to
%protocols where the users may rely on past messages they have individually
%transmitted and received with all other users in addition to the previous
%messages exchanged over the link in question to determine the prefix-free
%code for the next transmission. For instance, the prefix free code used
%for $M_{12,t}$ may be determined by each of $(M_{12}^{t-1},M_{13}^{t-1})$
%and $(M_{12}^{t-1},M_{23}^{t-1})$. This also means that, for the most
%general case discussed above where the code may be determined also by the
%input of the communicating users, we can obtain the same lower bounds for
%rates on the Alice-Charlie and Bob-Charlie links since one of the
%communicating parties, Charlie, has no input.  This also yields the same
%lower bound on the randomness $\rho$, leaving only the case of $R_{12}$
%open for now.}

\appendix
\section{Perfectly secure computation of {XOR} with zero-error}
Here we summarize the arguments of~\cite{DataPP14} specialized to perfectly secure computation of {XOR} (with zero error and perfect privacy), i.e., \eqref{eqn:ach_rate_1}-\eqref{eqn:ach_rate_4} hold with equality. We allow all input distributions $p_{XY}$ with full support. Alice and Bob each have a block $X^n$ and $Y^n$ of $n$ bits respectively, and Charlie wants to compute $Z^n$, component-wise {XOR} of the input bits. A simple protocol for this is: Alice samples $n$ i.i.d. uniformly distributed bits $(K_1,K_2,\hdots,K_n)$ from her private randomness and sends $M_{13}=K^n\oplus X^n$ to Charlie and $M_{12}=K^n$ to Bob. Bob computes $M_{12}\oplus Y^n$ and sends it to Charlie as $M_{23}$. Charlie computes $M_{13}\oplus M_{23}$, which is equal to $X^n\oplus Y^n$ and outputs it. Clearly, this protocol requires $n$ privately random bits as well as $n$ bits to be communicated on each of the three links. In Theorem~\ref{lem:xor_lbs}, we show that this is optimal.
\begin{lemma}\label{lem:cutset_indep}
Any perfectly secure protocol for computing {XOR} (with zero-error and perfect privacy), for $p_{XY}$ with full support, satisfies
\begin{align*}
&H(X^n|M_{12},M_{13}) = H(Y^n|M_{12},M_{23}) = 0, \\
I(M_{12}&;X^n,Y^n) = I(M_{13};X^n,Y^n) = I(M_{23};X^n,Y^n) = 0.
\end{align*}
\end{lemma}
\begin{proof}
See~\cite[Lemmas~2~and~3]{DataPP14}.
%Suppose, we have a secure zero-error protocol which results in a p.m.f. $p(x^n,y^n,m_{12},m_{23},m_{13},z^n)$.
%To prove $H(X^n|M)=0$, where $M=(M_{12},M_{13})$, suppose, to the contrary, that, $\exists m,x^n,\tilde{x}^n,x^n\neq \tilde{x}^n$ s.t. $p(x^n|m),p(\tilde{x}^n|m)>0${\color{red}, $p(m)>0$}. We will show that, for every $y^n$, $p(z^n|m,y^n)>0$ and $p(\tilde{z}^n|m,y^n)>0$, where $z^n=x^n\oplus y^n$ and $\tilde{z}^n=\tilde{x}^n\oplus y^n$, which contradicts our assumption that the protocol is zero-error. Suppose, $\exists y^n$, for which, $p(x^n|m,y^n)=0$. This implies that $p(x^n,y^n,m)=0$. Since $p(x^n,m)>0$, we have $p(y^n|x^n,m)=0$. Now, by privacy against Alice, $p(y^n|x^n,m)=p(y^n|x^n)\neq 0$, a contradiction. The proof of $H(Y^n|M_{12},M_{23})=0$ follows similarly.
%For $I(M_{12};X^n,Y^n)=0$, it is enough to prove that $I(M_{12};X^n)=0$. Consider any $x^n,\tilde{x}^n$ and $m_{12}$, we have $p(m_{12}|x^n)=p(m_{12}|x^n,y^n)=p(m_{12}|\tilde{x}^n,y^n)=p(m_{12}|\tilde{x}^n)$, where first and third equality follow from privacy against Alice and second equality follows from privacy against Bob. The proofs of $I(M_{13};X^n,Y^n)=0$ and $I(M_{23};X^n,Y^n)=0$ follow similarly.
\end{proof}
The lemma states that (i) examining the transcripts on the links which Alice is party to must reveal $X^n$ (similarly for Bob and $Y^n$), and (ii) examining the transcripts on any one of the links must reveal nothing about $X^n,Y^n$.
\begin{theorem}[Theorem~13 of~\cite{DataPP14}]\label{lem:xor_lbs}
Any perfectly secure protocol for computing {XOR} (with zero-error and perfect privacy), for $p_{XY}$ with full support, satisfies, 
%\[\mathbb{E}[L_{12}],\mathbb{E}[L_{23}],\mathbb{E}[L_{13}],\rho \geq n.\]
\[r_{12,n},r_{23,n},r_{13},\rho_n \geq 1.\]
\end{theorem}
\begin{proof}
We only prove the lower bound on $\mathbb{E}[L_{12}]$ and $\rho$. The others follow similarly.
We can lower bound $\mathbb{E}[L_{12}]$ by $H(M_{12})$ exactly as we did in the proof of converse (Section \ref{sec:converse}) of Theorem \ref{theorem:main_thm}. So
\begin{align}
nr_{12,n}=\mathbb{E}[L_{12}] &\geq H(M_{12}) \nonumber \\
&\geq H(M_{12}|M_{13}) \nonumber \\
&= H(M_{12},X^n|M_{13}) \label{eq:zeroerror_M12_1} \\
&\geq H(X^n|M_{13}) \nonumber \\
&= H(X^n), \label{eq:zeroerror_M12_2} 
\end{align}
where \eqref{eq:zeroerror_M12_1} and \eqref{eq:zeroerror_M12_2} follow from Lemma \ref{lem:cutset_indep}.

Now we apply the {\em distribution switching} idea from~\cite{DataPP14} to complete the argument. Briefly, we note that any secure protocol for {XOR}, where input distribution $p_{XY}$ has full support, continues to be a secure protocol even if we switch the input distribution to a different one $p_{\tilde{X}\tilde{Y}}$. This follows directly from zero-error and prefect privacy conditions. Together with Lemma \ref{lem:cutset_indep}, this implies that the marginal distributions of the transcripts $M_{12}$, $M_{23}$ and $M_{13}$ (and therefore their expected lengths) do not change if we switch the input distribution; see~\cite[Section~3.2]{DataPP14} for more details. This allows us to argue that 
\[nr_{12,n} \geq \sup_{p_{\tilde{X}\tilde{Y}}}H(\tilde{X}^n)=n,\] where $p_{\tilde{X}\tilde{Y}}$ is any distribution having full support. Now, taking the uniform distribution gives $r_{12,n} \geq 1$. For randomness,
\begin{align*}
n\rho_n &\geq H(M_{12},M_{23},M_{13}|X^n,Y^n) \\
&\geq H(M_{12}|X^n,Y^n) \\
&= H(M_{12}) \quad (\text{from Lemma }\ref{lem:cutset_indep})\\
&\geq H(M_{12}|M_{13}) \\
&\geq n. \quad (\text{as for }nr_{12,n} \text{ above})
\end{align*}
\end{proof}

\section*{Acknowledgment}
The work  was supported in part by the Bharti Centre for Communication, IIT Bombay, a grant from the Information Technology Research Academy, Media Lab Asia, to IIT Bombay and TIFR, a grant from the Department of Science and Technology, Government of India, to IIT Bombay, and a Ramanujan Fellowship from the Department of Science and Technology, Government of India, to V. Prabhakaran.

\end{document}